\documentclass{article}
\usepackage{arxiv}

\usepackage{amsthm}
\usepackage{amsmath}
\usepackage{amssymb}
\usepackage{aliascnt}
\usepackage{caption}
\usepackage[list=true,listformat=simple]{subcaption}
\usepackage{algorithm}
\usepackage{algpseudocode}
\usepackage{hhline}
\usepackage{makecell}
\usepackage{multirow}
\usepackage[figuresright]{rotating}
\usepackage{diagbox}
\usepackage{pgfplots}
\usepackage{comment}
\usepackage{pdfpages}
\usepackage{hyperref}
\usepackage{tabularx}
\usepackage{todonotes}
\usepackage{aliascnt}
\makeatletter
\let\corollary\@undefined
\let\endcorollary\@undefined
\let\lemma\@undefined
\let\endlemma\@undefined
\let\observation\@undefined
\let\endobservation\@undefined
\makeatother
\newtheorem{theorem}{Theorem}
\newaliascnt{lemma}{theorem}
\newtheorem{lemma}[lemma]{Lemma}
\aliascntresetthe{lemma}

\newaliascnt{corollary}{theorem}
\newtheorem{corollary}[corollary]{Corollary}
\aliascntresetthe{corollary}

\newaliascnt{observation}{theorem}
\newtheorem{observation}[observation]{Observation}
\aliascntresetthe{observation}

\newaliascnt{definition}{theorem}
\newtheorem{definition}[definition]{Definition}
\aliascntresetthe{definition}

\title{Faster FPT Algorithm for \textsc{5-Path Vertex Cover}
\thanks{An extended abstract of this paper appeared in the \href{http://www.dagstuhl.de/dagpub/978-3-95977-117-7}{Proceedings of the 44th International Symposium on Mathematical Foundations of Computer Science, 2019}~\cite{CervenyS19}. This version significantly reduces the number of rules needed from 51 to 16.}}

\renewcommand{\shorttitle}{Faster FPT Algorithm for \textsc{5-Path Vertex Cover}}

\author{Radovan~\v Cerven\'{y}\medskip\href{https://orcid.org/0000-0003-4528-9525}{\includegraphics[scale=0.06]{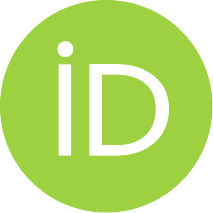}}\thanks{Supported by the Grant Agency of the Czech Technical University
in Prague, grant \mbox{No.~SGS20/208/OHK3/3T/18}.} \\
Department of Theoretical Computer Science\\
Faculty of Information Technology\\
Czech Technical University in Prague\\
Prague, Czech~Republic \\
\texttt{radovan.cerveny@fit.cvut.cz} \\
\And
Ond\v rej~Such\'{y}\medskip\href{https://orcid.org/0000-0002-7236-8336}{\includegraphics[scale=0.06]{orcid.pdf}}\thanks{The author acknowledges the support of the OP VVV MEYS funded project
CZ.02.1.01/0.0/0.0/16\_019/0000765 ``Research Center for Informatics''.} \\
Department of Theoretical Computer Science\\
Faculty of Information Technology\\
Czech Technical University in Prague\\
Prague, Czech~Republic \\
\texttt{ondrej.suchy@fit.cvut.cz} \\
}

\begin{document}
\maketitle

\begin{abstract}
The problem of \textsc{$d$-Path Vertex Cover, $d$-PVC} lies in determining a subset~$F$ of vertices of a~given graph $G=(V,E)$ such that $G \setminus F$ does not contain a~path on $d$ vertices. The paths we aim to cover need not to be induced. It is known that the \textsc{$d$-PVC} problem is NP-complete for any $d \ge 2$. When parameterized by the size of the solution $k$, \textsc{5-PVC} has direct trivial algorithm with $\mathcal{O}(5^kn^{\mathcal{O}(1)})$ running time and, since \textsc{$d$-PVC} is a special case of \textsc{$d$-Hitting Set}, an algorithm running in $\mathcal{O}(4.0755^kn^{\mathcal{O}(1)})$ time is known. In this paper we present an iterative compression algorithm that solves the \textsc{5-PVC} problem in $\mathcal{O}(4^kn^{\mathcal{O}(1)})$ time.
\end{abstract}

\newcounter{arulecounter}
\setcounter{arulecounter}{-1}
\renewcommand{\thearulecounter}{(R\arabic{arulecounter})}
\newcounter{asubrulecounter}[arulecounter]
\renewcommand{\theasubrulecounter}{(R\arabic{arulecounter}.\arabic{asubrulecounter})}
\newcounter{asubsubrulecounter}[asubrulecounter]
\renewcommand{\theasubsubrulecounter}{(R\arabic{arulecounter}.\arabic{asubrulecounter}\alph{asubsubrulecounter})}

\newcounter{globalcounter}
\setcounter{globalcounter}{0}

\newaliascnt{theoremcounter}{globalcounter}
\newaliascnt{corollarycounter}{globalcounter}
\newaliascnt{lemmacounter}{globalcounter}
\newaliascnt{observationcounter}{globalcounter}
\newaliascnt{definitioncounter}{globalcounter}

\theoremstyle{definition}

\newtheorem{areductionrule}[arulecounter]{Reduction rule}
\newtheorem{areductionsubrule}[asubrulecounter]{Reduction rule}
\newtheorem{areductionsubsubrule}[asubsubrulecounter]{Reduction rule}

\newtheorem{abranchingrule}[arulecounter]{Branching rule}
\newtheorem{abranchingsubrule}[asubrulecounter]{Branching rule}
\newtheorem{abranchingsubsubrule}[asubsubrulecounter]{Branching rule}

\newenvironment{correctness}{\noindent\textit{Proof of correctness.}}{\null~\hfill~$\Join$}

\newcommand{\arulecounterautorefname}{Rule}
\newcommand{\asubrulecounterautorefname}{Rule}
\newcommand{\asubsubrulecounterautorefname}{Rule}

\renewcommand{\sectionautorefname}{Section}
\renewcommand{\subsectionautorefname}{Subsection}

\newcommand{\ostar}{\mathcal{O}^*}

\usetikzlibrary{patterns}

\definecolor{songblue}{RGB}{112, 243, 255}
\definecolor{songred}{RGB}{255, 102, 71}
\definecolor{songblack}{RGB}{0,0,0}
\tikzset{songbluenode/.style={circle,fill=songblue,draw=songblack, inner sep=3pt},
         songbluenodedashed/.style={circle,preaction={fill, white}, fill=songblue,draw=songblack, dashed, pattern=north west lines, pattern color=songblue, inner sep=3pt},
         songrednode/.style={circle,fill=songred,draw=songblack, inner sep=3pt},
         songrednodedashed/.style={circle,preaction={fill, white}, fill=songred,draw=songblack, dashed, pattern=north west lines, pattern color=songred, inner sep=3pt},
         songblacknode/.style={circle,fill=songblack,draw=songblack, inner sep=3pt},
         songblacknodedashed/.style={circle,preaction={fill, white}, fill=songblack,draw=songblack, dashed, pattern=north west lines, pattern color=songblack, inner sep=3pt},
         songblackpath/.style={draw=songblack, line width=0.5pt, -},
         songblackpathdashed/.style={draw=songblack, dashed, line width=0.5pt, -},
            }

\section{Introduction}
The problem of \textsc{$d$-Path Vertex Cover,~$d$-PVC} lies in determining a~subset~$F$ of vertices of a~given graph $G=(V,E)$ such that $G \setminus F$ does not contain a~path on~$d$ vertices (even not a non-induced one). The problem was first introduced by Brešar et al.~\cite{BresarKKS11}, but its NP-completeness for any $d \ge 2$ follows already from the meta-theorem of Lewis and Yannakakis~\cite{LewisY80}. The \textsc{2-PVC} problem corresponds to the well known \textsc{Vertex Cover} problem and the \textsc{3-PVC} problem is also known as \textsc{Maximum Dissociation Set}. The \textsc{$d$-PVC} problem is motivated by the field of designing secure wireless communication protocols~\cite{Novotny10} or in route planning and speeding up shortest path queries~\cite{FunkeNS16}.

Since the problem is NP-hard, any algorithm solving the problem exactly is expected to have exponential running time. If one measures the running time solely in terms of the input size, then several efficient (faster than trivial enumeration) exact algorithms are known for \textsc{2-PVC} and \textsc{3-PVC}. In particular, \textsc{2-PVC} (\textsc{Vertex Cover}) can be solved in $\mathcal{O}(1.1996^n)$ time and polynomial space due to Xiao and Nagamochi~\cite{XiaoN17} and \textsc{3-PVC} can be solved in $\mathcal{O}(1.4656^n)$ time and polynomial space due to Xiao and Kou~\cite{XiaoK17exact}.

In this paper we aim on the parameterized analysis of the problem, that is, to confine the exponential part of the running time to a specific parameter of the input, presumably much smaller than the input size. The problem is called \emph{fixed-parameter tractable} fi there exists an algorithm (called a \emph{fixed-parameter algorithm}) that runs in $f(k)n^{O(1)}$ time, where $k$ is the parameter. The class of problems containing all fixed-parameter tractable problems is called \emph{FPT}. See Cygan et al.~\cite{CyganFKLMPPS15} for a broader introduction to parameterized algorithms.

When parameterized by the size of the solution~$k$, the \textsc{$d$-PVC} problem is directly solvable by a~trivial FPT algorithm that runs in $\ostar(d^k)$ time.\footnote{The $\ostar()$ notation suppresses all factors polynomial in the input size.} However, since \textsc{$d$-PVC} is a special case of \textsc{$d$-Hitting Set}, it was shown by Fomin et al.~\cite{FominGKLS10} that for any $d\ge4$ we have an algorithm solving \textsc{$d$-PVC} in $\ostar((d - 0.9245)^{k})$. In order to find more efficient solutions, the problem has been extensively studied in a~setting where~$d$ is a~small constant. For the \textsc{2-PVC} (\textsc{Vertex Cover}) problem, the algorithm of Chen, Kanj, and Xia~\cite{ChenKX10} has the currently best known running time of $\ostar(1.2738^k)$. For \textsc{3-PVC}, Tu~\cite{Tu15} used iterative compression to achieve a~running time $\ostar(2^k)$. This was later improved by Katrenič~\cite{Katrenic16} to $\ostar(1.8127^k)$, by Xiao and Kou~\cite{XiaoK17} to $\ostar(1.7485^k)$ by using a~branch-and-reduce approach and it was further improved by Tsur~\cite{TSUR20191} to $\ostar(1.713^k)$. For the \textsc{4-PVC} problem, Tu and Jin~\cite{TuJ16} again used iterative compression and achieved a~running time $\ostar(3^k)$ and Tsur~\cite{Tsur21-4PVC} gave the current best algorithm that runs in $\ostar(2.619^k)$ time. For $d=5$, $6$, and $7$ Tsur~\cite{Tsur19b} claimed algorithms for \textsc{$d$-PVC} with running times $\ostar(3.945^k)$, $\ostar(4.947^k)$, and $\ostar(5.951^k)$, respectively. Recently, the authors of this paper claimed to have developed a procedure that generates even faster algorithms for \textsc{$d$-PVC} for some $d$ \cite{CervenyS21}.

In this paper, we present an algorithm that solves the \textsc{5-PVC} problem parameterized by the size of the solution~$k$ in $\ostar(4^k)$ time by employing the iterative compression technique. Using the result of Fomin et al.~\cite{FominGLS19} this also yields $\mathcal{O}(1.7501^n )$ time algorithm improving upon previously known $\mathcal{O}(1.7547^n)$ time algorithm.

\subparagraph{Organization of this paper.} We introduce the notation and define the \textsc{5-PVC} problem in \autoref{chapter:prelims}. Our disjoint compression routine for iterative compression is exposed in \autoref{chapter:5-PVCwB}.
We conclude this paper with a few open questions.

\section{Preliminaries}
\label{chapter:prelims}
We use the $\ostar$ \emph{notation} as described by Fomin and Kratsch~\cite{FominK10}, which is a~modification of the big-$\mathcal{O}$ notation suppressing all factors bounded by a polynomial of the input size. We use the notation of parameterized complexity as described by Cygan et al.~\cite{CyganFKLMPPS15}. We use standard graph notation and consider simple and undirected graphs unless otherwise stated. Vertices of graph~$G$ are denoted by $V(G)$, edges by $E(G)$. By $G[X]$ we denote the subgraph of~$G$ induced by vertices of $X \subseteq V(G)$. By $N(v)$ we denote the~set of neighbors of $v \in V(G)$ in~$G$. Analogically, $N(X) = \bigcup_{x \in X} N(x)$ denotes the~set of neighbors of vertices in $X \subseteq V(G)$. The degree of vertex~$v$ is denoted by $deg(v) = |N(v)|$. For simplicity, we write $G \setminus v$ for $v \in V(G)$ and $G \setminus X$ for $X \subseteq V(G)$ as shorthands for $G[V(G)\setminus \{v\}]$ and $G[V(G)\setminus X]$, respectively.

A~\emph{k-path}, denoted as an ordered~$k$-tuple $P_k = (p_1, p_2, \ldots, p_k)$, is a~path on $k$~vertices $\{p_1, p_2, \ldots, p_k\}$. A~path $P_k$ \emph{starts} at vertex~$x$ when $p_1 = x$. A~\emph{$k$-cycle} is a~cycle on~$k$ vertices. A~\emph{triangle} is a~3-cycle. A~\emph{$P_5$-free} graph is a~graph that does not contain a~$P_5$ as a~subgraph (the $P_5$ need not to be induced). The \textsc{5-Path Vertex Cover} problem is formally defined as follows:

\vspace{2mm}
\noindent
\begin{tabularx}{\textwidth}{|l|X|}
  \hline
\multicolumn{2}{|l|}{\textsc{5-Path Vertex Cover, 5-PVC}} \\ \hline
  \textsc{Input}: & A~graph $G=(V,E)$, an integer $k \in Z^{+}_0$. \\
  \textsc{Output}: & A~set $F \subseteq V$, such that $|F| \leq k$ and $G \setminus F$ is a~$P_5$-free graph. \\
  \hline
\end{tabularx}
\vspace{2mm}

\begin{definition}\label{def:star}
A~\emph{star} is a~graph~$S$ with vertices $V(S) = \{s\} \cup \{l_1,\ldots,l_k\}$, $k \ge 3$ and edges $E(S) = \{ \{s,l_i\} \mid i \in \{1,\ldots,k\} \}$. Vertex~$s$ is called \emph{a~center}, vertices $L = \{l_1,\ldots,l_k\}$ are called \emph{leaves}.
\end{definition}

\begin{definition}\label{def:star_tri}
A~\emph{star with a~triangle} is a~graph $S^\triangle$ with vertices $V(S^\triangle) = \{s,t_1,t_2\}$ $\cup$ $\{l_1,\ldots,l_k\}$, $k \ge 1$ and edges $E(S^\triangle) = \{ \{s,t_1\}, \{s,t_2\},$ $\{t_1,t_2\} \}$ $\cup$ $\{ \{s,l_i\} \mid i \in \{1,\ldots,k\} \}$. Vertex~$s$ is called \emph{a~center}, vertices $T = \{t_1, t_2\}$ are called \emph{triangle vertices} and vertices $L = \{l_1,\ldots,l_k\}$ are called \emph{leaves}.
\end{definition}

\begin{definition}\label{def:distar}
A~\emph{di-star} is a~graph~$D$ with vertices $V(D) = \{s,s'\}$ $\cup$ $\{l_1,\ldots,l_k\}$ $\cup$ $\{l'_1,\ldots,l'_m\}$, $k \ge 1, m \ge 1$ and edges $E(D) = \{ \{s,s'\} \}$ $\cup$ $\{ \{s,l_i\} \mid i \in \{1,\ldots,k\} \}$ $\cup$ $\{ \{s',l'_j\} \mid j \in \{1,\ldots,m\} \}$. Vertices $s,s'$ are called \emph{centers}, vertices $L=\{l_1,\ldots,l_k\}$ and $L' = \{l'_1,\ldots,l'_m\}$ are called \emph{leaves}.
\end{definition}

\begin{lemma}\label{lemma:p5-free-graphs}
If a~connected graph is $P_5$-free and has more than 5 vertices, then it is a~star, a~star with a~triangle, or a~di-star.
\end{lemma}

\begin{proof}
Suppose we have a~$P_5$-free graph~$G$ on at least 5 vertices. Firstly,~$G$ does not contain a~$k$-cycle, $k\ge5$ as a~subgraph, since $P_5$ is a~subgraph of such a~$k$-cycle. Secondly,~$G$ does not contain a~4-cycle as a~subgraph, since~$G$ has at least 5 vertices and it is connected which implies that there is at least one vertex connected to the 4-cycle which in turn implies a~$P_5$ in~$G$. Finally,~$G$ does not contain two edge-disjoint triangles as a~subgraph, since~$G$ is connected, the two triangles are either sharing a~vertex or are connected by some path, which in both cases implies a~$P_5$ in~$G$. Consequently,~$G$ contains either exactly one triangle or is acyclic.

Consider the first case where~$G$ contains exactly one triangle. Label the vertices of the triangle with $\{t_1,t_2,t_3\}$. Then we claim that all vertices outside the triangle are connected by an edge to exactly one vertex of that triangle, let that vertex be $t_1$. Indeed, for contradiction suppose they are not. Since we have at least 5 vertices in~$G$, label the two existing vertices outside the triangle~$x$ and~$y$. Then we either have~$x$ and~$y$ connecting to two different vertices of the triangle, let them be $t_1, t_2$, which immediately implies a~$P_5 = (x,t_1,t_3,t_2,y)$ in~$G$, or we have a~$P_3 = (x,y,t_1)$ connected to the triangle, which again implies a~$P_5 = (x,y,t_1,t_2,t_3)$. Hence, if~$G$ contains a~triangle, then it is a~star with a~triangle.

Consider the second case where~$G$ is acyclic. Then we claim that there is a~dominating edge in~$G$, i.e. an edge $e = \{x,y\}$ such that $V(G) = N(\{x,y\})$. Indeed, for contradiction suppose that there is no such edge. Then we have that for each edge $e=\{x,y\}$ in~$G$ there must be a~vertex~$v$ that is adjacent neither to~$x$, nor to~$y$. Assume that~$v$ is connected to~$y$ through some vertex~$u$. The same also holds for the edge $\{y,u\}$, so assume that there is a~vertex $v' \ne x$ that is connected to~$u$ through some vertex $u' \ne y$. But then we have a~$P_5 = (x,y,u,u',v')$ in~$G$.

Label the dominating edge $e=\{s, s'\}$. Here, if only one of the vertices $s, s'$ has degree greater than one, we have a~star, otherwise we have a~di-star.
\end{proof}

\section{5-PVC with $P_5$-free bipartition}
\label{chapter:5-PVCwB}

We employ the generic iterative compression framework as described by Cygan et al.~\cite[pages 80--81]{CyganFKLMPPS15}. We skip the generic steps and only present the disjoint compression routine (see also \autoref{section:final-remarks} for a brief discussion of the whole iterative compression algorithm).
That is, we assume that  we are given a solution to the problem and search for another solution which is strictly smaller than and disjoint from the given one. Moreover, if the graph induced by the given solution contains a $P_5$, then we can directly answer no.
Hence our routine \textsc{disjoint\_r} restricts itself to a problem called \textsc{5-PVC with $P_5$-free Bipartition} and we need it to run in $\ostar(3^k)$ time.

A~\emph{$P_5$-free bipartition} of graph $G=(V,E)$ is a~pair $(V_1, V_2)$ such that $V = V_1 \cup V_2$, $V_1 \cap V_2 = \emptyset$ and $G[V_1], G[V_2]$ are $P_5$-free. The \textsc{5-PVC with $P_5$-free Bipartition} problem is formally defined as follows:

\vspace{2mm}
\noindent
\begin{tabularx}{\textwidth}{|l|X|}
  \hline
\multicolumn{2}{|l|}{\textsc{5-PVC with $P_5$-free Bipartition, 5-PVCwB}} \\ \hline
  \textsc{Input}: & A~graph $G=(V,E)$ with $P_5$-free bipartition $(V_1, V_2)$, an integer $k \in Z^{+}_0$. \\
  \textsc{Output}: & A~set $F \subseteq V_2$, such that $|F| \leq k$ and $G \setminus F$ is a~$P_5$-free graph. \\
  \hline
\end{tabularx}
\vspace{2mm}

\noindent
Throughout this paper the vertices from $V_1$ will be also referred to as ``red'' vertices and vertices from $V_2$ will be also referred to as ``blue'' vertices. The same colors will also be used in figures with the same meaning.

\subsection{Algorithm Outline}
Our algorithm is a~recursive procedure $\textsc{disjoint\_r}(G,V_1,V_2,F,k)$, where~$G$ is the input graph, $V_1,V_2$ are the partitions of the $P_5$-free bipartition of~$G$,~$F$ is the solution being constructed, and~$k$ is the maximum number of vertices we can still add to~$F$.
The procedure repeatedly tries to apply a~series of rules with a~condition that a~rule $(RI)$ can be applied only if all rules that come before $(RI)$ cannot be applied. It is paramount that in every call of $\textsc{disjoint\_r}$ at least one rule can be applied. The main work is done in rules of two types: \emph{reduction rules} and \emph{branching rules}.

A~\emph{reduction rule} is used to simplify a~problem instance, i.e. remove some vertices or edges from~$G$ and possibly add some vertices to a~solution, or to halt the algorithm. A~\emph{branching rule} splits the problem instance into at least two subinstances. The branching is based on subsets of vertices that we try to add to a~solution and by adding them to the solution we also remove them from~$G$.

The notation we use to denote the individual branches of a~branching rule is as follows: $\langle\, X_1 \mid X_2 \mid \ldots \mid X_l \,\rangle$. Such a~rule has~$l$ branches and $X_1,X_2,\ldots,X_l$ are subsets of $V_2$ which we try to add to the solution. This rule is translated into the following~$l$ calls of the procedure: $$\textsc{disjoint\_r}(G \setminus X_i ,V_1, V_2 \setminus X_i, F \cup X_i, k - |X_i|)~\textup{for}~i \in \{1,\ldots,l\}$$

A~rule is \emph{applicable} if the conditions of the rule are satisfied and none of the previous rules is applicable.

A~reduction rule is \emph{correct} if it satisfies that the problem instance has a~solution if and only if the simplified problem instance has a~solution. A~branching rule is \emph{correct} if it satisfies that if the problem instance has a~solution, then at least one of the branches of the rule will return a~solution.

When we say we \emph{delete} a~vertex, we mean that we remove it from~$G$ and also add it to the solution~$F$. When we say we \emph{remove} a~vertex, we mean that we remove it from~$G$ and \emph{do not} add it to the solution~$F$.

For the rest of this paper assume that the parameters of the current call of $\textsc{disjoint\_r}$ are $G,V_1,V_2,F,k$.
\subsection{Preprocessing}

\begin{areductionrule}\label{rule:R0}
This rule stops the recursion of $\textsc{disjoint\_r}$. It has three stopping conditions:
\begin{enumerate}
  \item If $k < 0$, return $\emph{no solution}$;
  \item else if~$G$ is $P_5$-free, return~$F$;
  \item else if $k = 0$, return $\emph{no solution}$.
\end{enumerate}
\end{areductionrule}

\begin{areductionrule}\label{rule:R1}
Let $v \in V(G)$ be a vertex such that there is no $P_5$ in~$G$ that uses $v$. Then remove~$v$ from~$G$.
\end{areductionrule}

\begin{correctness}
Let $v \in V(G)$ be a vertex that is not used by any $P_5$ in~$G$ and let~$F$ be a~solution to the \textsc{5-PVCwB} instance $(G \setminus v, V_1 \setminus \{v\}, V_2 \setminus \{v\}, k)$. Then~$F$ is also a~solution to $(G,V_1,V_2,k)$ since $v$ is not used by any $P_5$ in~$G$.

If $(G \setminus v, V_1 \setminus \{v\}, V_2 \setminus \{v\}, k)$ does not have a~solution, then we claim that $(G,V_1,V_2,k)$ also does not have a~solution. Indeed, adding vertices can only create new $P_5$ paths.
\end{correctness}

\begin{areductionrule}\label{rule:domination}
 Let $u,v \in V_2$ be two vertices such that $u$ is contained in every $P_5$ in $G$ which contains $v$. Then move the vertex $v$ to $V_1$ (make it red).
\end{areductionrule}

\begin{correctness}
Let $u,v \in V_2$ be two vertices such that $u$ is contained in every $P_5$ in $G$ which contains $v$ and let~$F$ be a~solution to the \textsc{5-PVCwB} instance $(G, V_1 \setminus \{v\}, V_2 \cup \{v\}, k)$. Then~$F$ is also a~solution to $(G,V_1,V_2,k)$. Similarly, if~$F$ is a~solution to $(G,V_1,V_2,k)$ which does not contain $v$, then it is also a solution to $(G, V_1 \setminus \{v\}, V_2 \cup \{v\}, k)$.
Finally, if $F$ is a solution to $(G,V_1,V_2,k)$ which contains $v$, then $F'=(F \setminus v) \cup \{u\}$ is also a solution to $(G,V_1,V_2,k)$: Indeed, as all $P_5$ paths that contain $v$ also, by assumption, contain $u$, they are all covered by $F'$.
\end{correctness}

\begin{abranchingrule}\label{rule:R2}
Let~$P$ be a~$P_5$ in~$G$ with $X = V(P) \cap V_2$ such that $|X| \leq 3$. Then branch on $\langle\, x_1 \mid x_2 \mid \ldots \,\rangle, x_i \in X$, i.e. branch on the blue vertices of~$P$.
\end{abranchingrule}

\begin{correctness}
We have to delete at least one blue vertex in~$P$, thus branching on the blue vertices of~$P$ is correct.
\end{correctness}

\begin{lemma}\label{lemma:2}
Assume that Rules \ref{rule:R0} -- \ref{rule:R2} are not applicable. Then for each vertex $v \in V(G)$ there exists a $P_5$ in $G$ that uses $v$; every $P_5$ in~$G$ uses exactly one red vertex; and there are only isolated vertices in $G[V_1]$.
\end{lemma}

\begin{proof}
If \autoref{rule:R1} is not applicable, then for each vertex $v \in V(G)$ there exists a $P_5$ in $G$ that uses $v$. If \autoref{rule:R2} is not applicable, then every $P_5$ in~$G$ uses at most one red vertex and since $(V_1, V_2)$ is a~$P_5$-free bipartition we cannot have a~$P_5$ in~$G$ that uses no red vertex.

To prove that there are only isolated vertices in $G[V_1]$, assume for contradiction that there is an edge $e = \{v_1, v_2\}$ in $G[V_1]$. Since \autoref{rule:R1} and \autoref{rule:R2} are not applicable, there must be a path $P_1 = (u_1, u_2, u_3, u_4, u_5)$ in $G$ that uses $v_1$ and not uses $v_2$ and $P_2 = (w_1, w_2, w_3, w_4, w_5)$ in $G$ that uses $v_2$ and not uses $v_1$. Paths $P_1$ and $P_2$ are not necessarily disjoint.

Now consider the following cases for $v_1$ and $P_1$. If $v_1 = u_1$, then there is a path $P' = (v_2, v_1, u_2, u_3, u_4)$ contradicting \autoref{rule:R2} not being applicable. Similarly, if $v_1 = u_2$, then there is a path $P' = (v_2, v_1, u_3, u_4, u_5)$ contradicting \autoref{rule:R2} not being applicable. The same arguments apply for the cases where $v_1 = u_5$ and $v_1 = u_4$ respectively and the same logic applies also when considering $v_2$ and $P_2$.

Thus we have that $v_1 = u_3$ and $v_2 = w_3$. Now it suffices to see that either $w_2 \notin \{u_1,u_2\}$ or $w_2 \notin \{u_4, u_5\}$.  In the first case we get a path $P' = (u_1, u_2, v_1, v_2, w_2)$ and in the second case we get a path $P' = (u_5, u_4, v_1, v_2, w_2)$. In both cases we get a contradiction with \autoref{rule:R2} not being applicable.
\end{proof}

\subsection{Dealing with isolated vertices in $G[V_2]$}

\begin{lemma}\label{lemma:3}
Assume that Rules \ref{rule:R0} -- \ref{rule:R2} are not applicable. Let~$v$ be an isolated vertex in $G[V_2]$ and let~$F$ be a~solution to \textsc{5-PVCwB} which uses vertex~$v$. Then there exists a~solution $F'$ that does not use vertex~$v$ and $|F'| \leq |F|$.
\end{lemma}

\begin{proof}
From \autoref{lemma:2} we get that each $P_5$ in~$G$ which contains~$v$ must also start in~$v$, otherwise it would imply a~$P_5$ that uses more than one red vertex. In particular, $v$ has at most one red neighbor. Suppose that there exists a~path $P = (v,w,x,y,z)$ where~$w$ is a~red vertex and $\{x,y,z\} \cap F = \emptyset$ (see \autoref{fig:rule:R3}). If there is no such~$P$, then, as each vertex is in at least one $P_5$ due to \autoref{rule:R1} not being applicable, we have that each $P_5$ starting in~$v$ has at least one of the vertices $x,y,z$ in~$F$. In that case, we can put $F' = F \setminus \{v\}$ and the lemma~holds.

There cannot exist another path $P' = (v,w,x',y',z')$ such that $x' \neq x$ and $\{x',y',z'\} \cap F = \emptyset$, otherwise we would have a~$P_5 = (x', w, x, y, z)$ in~$G$ that is not hit by~$F$. Consequently, each $P_5$ that is hit only by vertex~$v$ also contains vertex~$x$, which implies that $F' = (F \setminus \{v\}) \cup \{x\}$ is a~solution and $|F'| \leq |F|$, thus the lemma~holds.
\end{proof}

\begin{figure}[t]
\centering
\begin{tikzpicture}[scale=1., every node/.style={transform shape}]
  \path[songblackpath] (0,0) -- (0.707,0.707) -- (2*0.707,0) -- (1+2*0.707,0) -- (2+2*0.707,0) {};
  \node[songbluenode, label=below:{$v$}] at (0,0) {};
  \node[songrednode,  label=above:{$w$}] at (0.707,0.707) {};
  \node[songbluenode, label=below:{$x$}] at (2*0.707,0) {};
  \node[songbluenode, label=below:{$y$}] at (1+2*0.707,0) {};
  \node[songbluenode, label=below:{$z$}] at (2+2*0.707,0) {};
\end{tikzpicture}
\caption{Configuration in  \autoref{lemma:3}.}
\label{fig:rule:R3}
\end{figure}

%

\begin{areductionrule}\label{rule:blue_isolated}\label{rule:R3}
 Let~$v$ be an isolated vertex in $G[V_2]$. Then move $v$ to $V_1$.
\end{areductionrule}

\begin{correctness}
 Let~$v$ be an isolated vertex in $G[V_2]$. If~$F$ is a~solution to the \textsc{5-PVCwB} instance $(G, V_1 \setminus \{v\}, V_2 \cup \{v\}, k)$, then~$F$ is also a~solution to $(G,V_1,V_2,k)$. Similarly, if~$F$ is a~solution to $(G,V_1,V_2,k)$, then by \autoref{lemma:3} we can assume that it does not contain $v$ and, hence, it is also a solution to $(G, V_1 \setminus \{v\}, V_2 \cup \{v\}, k)$.
\end{correctness}

\begin{observation}\label{lemma:no-isolated-vertex}
Assume that Rules \ref{rule:R0} -- \ref{rule:blue_isolated} are not applicable. Then there are no isolated vertices in $G[V_2]$.
\end{observation}


\subsection{Dealing with isolated edges in $G[V_2]$}

\begin{lemma}\label{lemma:6}
Assume that Rules \ref{rule:R0} -- \ref{rule:blue_isolated} are not applicable. Let~$v$ be a~blue vertex to which at least two red vertices are connected and let $C_v$ be a~connected component of $G[V_2]$ which contains~$v$. Then for each red vertex~$w$ connected to~$v$ we have that $N(w) \subseteq V(C_v)$.
\end{lemma}

\begin{proof}
Let $w_1, w_2$ be red vertices connected to~$v$. For contradiction assume that $w_1$ is connected to some vertex $v'$ in $G[V_2]$ such that $v' \notin V(C_v)$. From \autoref{lemma:no-isolated-vertex} we know that $v'$ has degree at least one in $G[V_2]$. Label some neighbor of $v'$ in $G[V_2]$  as $u'$. We obtained a~$P_5 = (u', v', w_1,v, w_2)$ which contradicts \autoref{lemma:2}.
\end{proof}

\begin{lemma}\label{lemma:6'}
Assume that Rules \ref{rule:R0} -- \ref{rule:blue_isolated} are not applicable. Let $e = \{u,v\} \subseteq V_2$ be a~blue edge to which at least two red vertices are connected ($|N(e) \cap V_1| \ge 2$) in a~way that to both~$u$ and~$v$ there is at least one red vertex connected ($|N(u) \cap V_1| \ge 1$, $|N(v) \cap V_1| \ge 1$). Let $C_e$ be a~connected component of $G[V_2]$ which contains~$e$. Then for each red vertex~$w$ connected to~$e$ we have that $N(w) \subseteq V(C_e)$.
\end{lemma}

\begin{proof}
Let $w_1, w_2$ be red vertices connected to~$e$ and assume that $w_1$ is connected to~$u$ and $w_2$ is connected to~$v$. For contradiction assume that $w_1$ is connected to some vertex $v'$ in $G[V_2]$ such that $v' \notin V(C_e)$. We obtain a~$P_5 = (v', w_1, u, v, w_2)$ which contradicts \autoref{lemma:2}.
\end{proof}

\begin{lemma}\label{lemma:7}
Let~$X$ be a~subset of $V_2$ such that $N(X) \cap V_1 = \emptyset$ and $|N(X) \cap V_2| = 1$.
Then \autoref{rule:domination} applies.
In particular, if $v \in V_2$ has degree one, then its neighbor is in $V_1$, or \autoref{rule:domination} applies.
\end{lemma}

\begin{proof}
Assume that $N(X) \cap V_2 = \{v\}$. Then each $P_5$ that uses some vertex in~$X$ must also use vertex~$v$, otherwise it would be contained in~$X$ which contradicts $G[V_2]$ being $P_5$-free.
\end{proof}

\begin{definition}
We say that two nodes $x,y$ are \emph{twins} if $N(x) \setminus \{y\} = N(y) \setminus \{x\}$.
\end{definition}

\begin{lemma}\label{lemma:twi}
Let $x,y$ be blue vertices that are twins. Let~$F$ be a~solution and $x \in F$. Then at least one of the following holds:
\begin{enumerate}
  \item[(1)] $y \in F$,
  \item[(2)] $F' = (F \setminus \{x\}) \cup \{y\}$ is a~solution.
\end{enumerate}
\end{lemma}

\begin{proof}
Assume that $x \in F$ and $y \notin F$. Since $x,y$ are twins, for each path $P = (p_1,p_2,p_3,p_4,p_5)$ with $p_i = x$ and $y \notin P$, there also exists a~path $P' = (p'_1,p'_2,p'_3,p'_4,p'_5)$ such that $p'_j = p_j$ for $j \in (\{1,2,3,4,5\} \setminus \{i\})$ and $p'_i = y$. Firstly, if there is no $P_5$ containing~$x$, then trivially \emph{(2)} holds. Secondly, if all $P_5$ paths that contain~$x$ are hit by some other vertex $z, z \neq x, z \in F$, then again \emph{(2)} holds. So suppose that there exists a~$P_5$ path~$P$ that is hit only by~$x$. If $y \notin P$, then we know that there is a~path $P'$ as described above and we get a~contradiction with~$F$ being a~solution since $P'$ is not hit by~$F$ and \emph{(1)} must hold. Otherwise, all $P_5$ paths that contain~$x$ also contain~$y$ and \emph{(2)} holds.
\end{proof}

\begin{lemma}\label{lemma:isolated_edge_situation}
 Assume that Rules \ref{rule:R0} -- \ref{rule:blue_isolated} are not applicable. If there is an isolated edge $e = \{u, v\}$ in $G[V_2]$, then there is exactly one red vertex $w$ connected to $e$. Moreover $N(w) \setminus \{u, v\} \neq \emptyset$.
\end{lemma}

\begin{proof}
 For contradiction assume that Rules \ref{rule:R0} -- \ref{rule:blue_isolated} are not applicable and there is an isolated edge $e = \{x,y\}$ in $G[V_2]$. If there is no $P_5$ that uses vertices from~$e$, then \autoref{rule:R1} is applicable on~$e$. Hence there are red vertices connected to~$e$. If there are at least two red vertices connected to~$e$, then from Lemmata~\ref{lemma:6} and \ref{lemma:6'} we know that these red vertices are not connected to any other vertices outside~$e$ and there again cannot be a~$P_5$ that uses vertices from~$e$ and \autoref{rule:R1} is applicable on~$e$. So, assume that there is exactly one red vertex connted to $e$ and label that vertex $w$. If $N(w) \subseteq \{u, v\}$, the vertices $u,v,w$ are not part of any $P_5$ and \autoref{rule:R1} is applicable on them. Therefore, $N(w) \setminus\{u,v\} \neq \emptyset$ and the lemma holds.
\end{proof}

\begin{abranchingrule}\label{rule:R4}
Let $e = \{u, v\}$ be an isolated edge in $G[V_2]$. Let there be a~red vertex~$w$ connected to at least one vertex in~$e$.  Assume that~$x$ is some vertex to which~$w$ connects outside~$e$ and let~$y$ be a~neighbor of~$x$ in $G[V_2]$. Then branch on $\langle\, v \mid x \mid y \,\rangle$.
\end{abranchingrule}

\begin{correctness}
By \autoref{lemma:isolated_edge_situation} $w$ is the only red vertex connected to $e$.
Firstly, assume that~$w$ is connected only to one vertex of~$e$. Then \autoref{rule:domination} applies.
Secondly, assume that~$w$ is connected to both vertices of~$e$. Since $u,v$ are twins, from \autoref{lemma:twi} it follows that we can try deleting only one of them. Thus branching on $\langle\, v \mid x \mid y \,\rangle$ is correct.
\end{correctness}

\begin{figure}[t]
\centering
\begin{tikzpicture}[scale=1., every node/.style={transform shape}]
  \path[songblackpath]  (0,0) -- (1,0) -- (1+0.707,0.707) -- (1+2*0.707,0) -- (2+2*0.707,0) {};
  \path[songblackpath]  (0,0) -- (1+0.707,0.707) {};

  \node[songbluenode, label=below:{$u$}] at (0,0) {};
  \node[songbluenode, label=below:{$v$}] at (1,0) {};
  \node[songrednode,  label=above:{$w$}] at (1+0.707,0.707) {};
  \node[songbluenode, label=below:{$x$}] at (1+2*0.707,0) {};
  \node[songbluenode, label=below:{$y$}] at (2+2*0.707,0) {};
\end{tikzpicture}
\caption{Configuration in Rule \ref{rule:R4}.}
\label{fig:rule:R4}
\end{figure}

\begin{observation}\label{lemma:no-isolated-edges}
Assume that Rules \ref{rule:R0} -- \ref{rule:R4} are not applicable. Then there are no isolated edges in $G[V_2]$.
\end{observation}

\begin{proof}
If there is an isolated edge in $G[V_2]$, then by \autoref{lemma:isolated_edge_situation} there is exactly one red vertex connected to it and \autoref{rule:R4} is applicable.
\end{proof}

\subsection{Dealing with isolated $P_3$ paths in $G[V_2]$}
\begin{lemma}\label{lemma:isolated_P3_situation}
 Assume that Rules \ref{rule:R0} -- \ref{rule:R4} are not applicable.
 Let~$P$ be a~$P_3 = (t,u,v)$ that forms a connected component in $G[V_2]$.
 There is only one red vertex~$w$ connected to~$P$.
 In particular, $w$ is connected to $t$ and $v$, to some component of $G[V_2]$ other than~$P$, and possibly to $u$.
\end{lemma}

\begin{proof}
 For contradiction assume that Rules \ref{rule:R0} -- \ref{rule:R4} are not applicable and there is an isolated $P_3$ path $P = (t,u,v)$ in $G[V_2]$. If there is no $P_5$ that uses vertices from~$P$, then \autoref{rule:R1} is applicable on~$P$.
 Hence there are red vertices connected to~$P_3$.
 Suppose there are at least two red vertices connected to~$P$. If they are connected to vertices $t,v$, then \autoref{rule:R2} is applicable, since there is a~$P_5$ that uses at least two red vertices. So suppose the red vertices are connected to a~single vertex or a~single edge in~$P$. Then from Lemmata~\ref{lemma:6} and \ref{lemma:6'} we know that those red vertices are not connected to any other vertices outside~$P$. Consequently, there cannot be a~$P_5$ that uses vertices from~$P$ and again \autoref{rule:R1} is applicable on~$P$.

 So suppose that there is a~$P_5$ that uses vertices from~$P$ and there is only one red vertex~$w$ connected to~$P$. If $w$ is not connected to $t$ or $v$, then \autoref{rule:domination} is applicable.
\end{proof}

\begin{abranchingrule}\label{rule:R5}
Let~$P$ be a~$P_3 = (t,u,v)$ that forms a connected component in $G[V_2]$ and $w$ be the only red vertex connected to~$P$. Assume that~$x$ is some vertex to which~$w$ connects outside~$P$ and let~$y$ be a~neighbor of~$x$ in $G[V_2]$ (see \autoref{fig:rule:R5}). Then branch on $\langle\, u \mid v \mid x \,\rangle$.
\end{abranchingrule}

\begin{correctness}
By \autoref{lemma:isolated_P3_situation}, vertex~$w$ is connected to $t,v$ in~$P$ and~$w$ can be also connected to~$u$ in~$P$.
If we do not delete vertex~$x$, then we have to delete something in~$P$. In both cases, when~$w$ is connected to~$u$ and when not, $t,v$ are twins and from \autoref{lemma:twi} we know that we have to try only one of $t,v$. Thus branching on $\langle\, u \mid v \mid x \,\rangle$ is correct.
\end{correctness}

\begin{figure}[t]
\centering
\begin{tikzpicture}[scale=1., every node/.style={transform shape}]
  \path[songblackpath]  (0,0) -- (1,0) -- (2,0)--(1,1) -- (1,2) -- (2,2)  {};
  \path[songblackpath] (0,0) --(1,1) {};
  \path[songblackpathdashed] (1,0) -- (1,1) {};
  \path[songblackpath] (2,0) --(1,1) {};

  \node[songbluenode, label=below:{$t$}] at (0,0) {};
  \node[songbluenode, label=below:{$u$}] at (1,0) {};
  \node[songbluenode, label=below:{$v$}] at (2,0) {};
  \node[songrednode,  label=above left:{$w$}] at (1,1) {};
  \node[songbluenode, label=above:{$x$}] at (1,2) {};
  \node[songbluenode, label=above:{$y$}] at (2,2) {};
\end{tikzpicture}
\caption{Configuration in Rule \ref{rule:R5}.}
\label{fig:rule:R5}
\end{figure}

\begin{observation}\label{lemma:no-isolated-p3-paths}
Assume that Rules \ref{rule:R0} -- \ref{rule:R5} are not applicable. Then there are no isolated $P_3$ paths in $G[V_2]$.
\end{observation}

\begin{proof}
For contradiction assume that Rules \ref{rule:R0} -- \ref{rule:R4} are not applicable and there is an isolated $P_3$ in $G[V_2]$. Then, by \autoref{lemma:isolated_P3_situation}, \ref{rule:R5} applies.
\end{proof}

\subsection{Dealing with isolated triangles in $G[V_2]$}
\begin{lemma}\label{lemma:isolated_triangle_situation}
Assume that Rules \ref{rule:R0} -- \ref{rule:R5} are not applicable.
Let~$T$ be a~$K_3 = \{t,u,v\}$ that forms a connected component in $G[V_2]$.
There is only one red vertex~$w$ connected to~$T$.
Furthermore, $w$ is connected to at least two vertices of $T$ and to some component of $G[V_2]$ other than~$T$.
\end{lemma}

\begin{proof}
 For contradiction assume that Rules \ref{rule:R0} -- \ref{rule:R5} are not applicable and there is an isolated triangle $T = \{t,u,v\}$ in $G[V_2]$. If there is no $P_5$ that uses vertices from~$T$, then \autoref{rule:R1} is applicable on~$T$.
 Hence there are red vertices connected to~$T$.
 Suppose there are at least two red vertices connected to~$T$. If the red vertices are not connected to a~single vertex in~$T$, then \autoref{rule:R2} is applicable, since there is a~$P_5$ that uses at least two red vertices. So suppose the red vertices are connected to a~single vertex in~$T$. Then from \autoref{lemma:6} we know that those red vertices are not connected to any other vertices outside~$T$. Consequently, there cannot be a~$P_5$ that uses vertices from~$T$ and again \autoref{rule:R1} is applicable on~$T$.

So suppose that there is a~$P_5$ that uses vertices from~$T$ and there is only one red vertex~$w$ connected to~$T$. If $w$ is connected to only one vertex of $T$, then \autoref{rule:domination} is applicable on the other vertices of $T$.
\end{proof}

\begin{abranchingrule}\label{rule:R6}
Let~$T$ be a~$K_3 = \{t,u,v\}$ that forms a connected component in $G[V_2]$ and $w$ the only red vertex connected to~$T$.
Suppose that vertex~$w$ is connected to at least two vertices in~$T$, let those vertices be $u,v$.
Assume that~$x$ is some vertex to which~$w$ connects outside~$T$ and let~$y$ be a~neighbor of~$x$ in $G[V_2]$ (see \autoref{fig:rule:R6}).
Then branch on $\langle\, t \mid v \mid x \,\rangle$.
\end{abranchingrule}

\begin{correctness}
If we do not delete vertex~$x$, then we have to delete something in~$T$. Since $u,v$ are twins, from \autoref{lemma:twi} we know that we have to try only one of $u,v$. Thus branching on $\langle\, t \mid v \mid x \,\rangle$ is correct.
\end{correctness}
\begin{figure}[t]
\centering
\begin{tikzpicture}[scale=1., every node/.style={transform shape}]
  \path[songblackpath]  (0,0) -- (-0.966,-0.259) -- (-0.259,-0.966) -- (0,0) {};
  \path[songblackpath] (0,0) -- (0.707,0.707) -- (2*0.707,0) -- (1+2*0.707,0) {};
  \path[songblackpathdashed] (-0.966,-0.259) -- (0.707,0.707)  {};
  \path[songblackpath] (-0.259,-0.966) -- (0.707,0.707)  {};

  \node[songbluenode, label=below:{$t$}] at (-0.966,-0.259) {};
  \node[songbluenode, label=below:{$u$}] at (-0.259,-0.966) {};
  \node[songbluenode, label=below left:{$v$}] at (0,0) {};
  \node[songrednode,  label=above:{$w$}] at (0.707,0.707) {};
  \node[songbluenode, label=below:{$x$}] at (2*0.707,0) {};
  \node[songbluenode, label=below:{$y$}] at (1+2*0.707,0) {};
\end{tikzpicture}
\caption{Configuration in Rule \ref{rule:R6}.}
\label{fig:rule:R6}
\end{figure}

\begin{observation}\label{lemma:no-isolated-triangles}
Assume that Rules \ref{rule:R0} -- \ref{rule:R6} are not applicable. Then there are no isolated triangles in $G[V_2]$.
\end{observation}

\begin{proof}
For contradiction assume that Rules \ref{rule:R0} -- \ref{rule:R5} are not applicable and there is an isolated triangle $T = \{t,u,v\}$ in $G[V_2]$. Then, by \autoref{lemma:isolated_triangle_situation}, \autoref{rule:R6} is applicable.
\end{proof}

\subsection{Dealing with 4-cycles in $G[V_2]$}
\begin{lemma}\label{lemma:delete}
Let~$C$ be a~connected component of $G[V_2]$ and $X = V(C) \cap N(V_1)$. Let~$F$ be a~solution that deletes at least $|X|$ vertices in~$C$. Then $F' = (F \setminus V(C)) \cup X$ is also a~solution and $|F'| \le |F|$.
\end{lemma}

\begin{proof}
Each $P_5$ that uses some vertex in~$C$ must also use some vertex $x \in X$, otherwise it would be contained in~$C$ which contradicts $G[V_2]$ being $P_5$-free. Consequently, any $P_5$ that is hit by a~vertex from~$C$ in the solution~$F$ can be also hit by some vertex $x \in X$ and thus $F' = (F \setminus V(C)) \cup X$ is also a~solution and $|F'| \leq |F|$.
\end{proof}

\begin{lemma}\label{lemma:4cycle_situation}
Assume that Rules \ref{rule:R0} -- \ref{rule:R6} are not applicable.
Let~$Q$ be a connected component in $G[V_2]$ which contains a~4-cycle.
Then $Q$ is a~subgraph of $K_4$ and there is exactly one red vertex connected to~$Q$ and it must be connected to at least two vertices in~$Q$ and to some component of $G[V_2]$ other than $Q$.
\end{lemma}

\begin{proof}
Assume that Rules \ref{rule:R0} -- \ref{rule:R6} are not applicable and there is a~component~$Q$ in $G[V_2]$ that contains a~4-cycle as a~subgraph, label the vertices of the 4-cycle $(v_1,v_2,v_3,v_4)$. Observe that~$Q$ is a~subgraph of $K_4$, as otherwise there would be a~$P_5$ in $G[V_2]$.

If there is no $P_5$ that uses vertices from~$Q$, then \autoref{rule:R1} is applicable on~$Q$. Hence there are red vertices connected to~$Q$. Suppose that there are at least two red vertices connected to~$Q$. If the red vertices are not connected to a~single vertex or a~single edge in~$Q$, then \autoref{rule:R2} is applicable, since there is a~$P_5$ that uses at least two red vertices. So suppose the red vertices are connected to a~single vertex or a~single edge in~$Q$. Then from Lemmata~\ref{lemma:6} and \ref{lemma:6'} we know that those red vertices are not connected to any other vertices outside~$Q$. Then, we have that every $P_5$, which uses some vertices from $Q$, actually uses all of the vertices from $Q$ and exactly one red vertex connected to $Q$. But then, \autoref{rule:domination} is applicable on $Q$. Hence, we have that there is exactly one red vertex connected to $Q$, label it $w$. Again, if $w$ is not connected to some component of $G[V_2]$ other than $Q$, then \autoref{rule:domination} applies on $Q$. Now assume, that $w$ is connected to only one vertex of $Q$. Then \autoref{rule:domination} is applicable on that vertex and the other vertices of $Q$. Therefore, vertex $w$ must be connected to at least two vertices of $Q$ and the lemma holds.

\end{proof}

Let~$Q$ be a connected component such that a~4-cycle is a~subgraph of~$Q$, label the vertices of the 4-cycle $(v_1,v_2,v_3,v_4)$. We will call pairs of vertices $\{v_1,v_3\}$ and $\{v_2,v_4\}$ \emph{diagonal}, all other pairs will be called \emph{non-diagonal}.




\begin{abranchingrule}\label{rule:R7.2.2}
Let~$Q$ be a connected component in $G[V_2]$ such that $Q \subseteq K_4$ and a~4-cycle $(v_1,v_2,v_3,v_4)$ is a~subgraph of~$Q$.
Assume that there is only one red vertex~$w$ connected to~$Q$ and $X = V(Q) \cap N(w)$.
Set~$X$ contains at least one diagonal pair, let that pair be $\{v_1,v_3\}$ (see \autoref{fig:rule:R7.2.2}). Then branch on $\langle\, v_1 \mid v_2 \mid v_4 \,\rangle$.
\end{abranchingrule}

\begin{correctness}
We have to delete something in~$Q$. Since $v_1,v_3$ are twins, from \autoref{lemma:twi} we know that we have to try only one of $v_1,v_3$. Thus branching on $\langle\, v_1 \mid v_2 \mid v_4 \,\rangle$ is correct.
\end{correctness}




\begin{abranchingrule}\label{rule:R7.2.5}
Let~$Q$ be a connected component in $G[V_2]$ such that $Q \subseteq K_4$ and a~4-cycle $(v_1,v_2,v_3,v_4)$ is a~subgraph of~$Q$.
Assume that there is only one red vertex~$w$ connected to~$Q$ and $X = V(Q) \cap N(w)$.
Suppose that set~$X$ is of size 2 and forms a non-diagonal pair, let that pair be $\{v_1,v_2\}$.
Assume that there is a vertex $x$ to which~$w$ connects outside~$Q$ and let~$y$ be a~neighbor of~$x$ in $G[V_2]$ (see \autoref{fig:rule:R7.2.5}). Then branch on $\langle\, \{v_1,v_2\} \mid x \mid y \,\rangle$.
\end{abranchingrule}

\begin{correctness}
If none of the vertices $x, y$ is deleted, then we have to delete at least two vertices in~$Q$. From \autoref{lemma:delete} we know that we only have to try deleting vertices $\{v_1,v_2\}$. Thus branching on $\langle\, \{v_1,v_2\} \mid x \mid y \,\rangle$ is correct.
\end{correctness}

\begin{figure}[t]
\centering
\begin{subfigure}{0.49\textwidth}
\centering
\begin{tikzpicture}[scale=1., every node/.style={transform shape}]
  \path[songblackpathdashed] (0,0) -- (1,-1) {};
  \path[songblackpathdashed] (0,-1) -- (1,0) {};
  \path[songblackpath]  (0,0) -- (1,0) -- (1,-1) -- (0,-1) -- (0,0) {};
  \begin{scope}[shift={(0,0)}]
  \path[songblackpathdashed] (0,-1) .. controls (-0.5,0) .. ({60:1}) {};
  \path[songblackpathdashed] (1,0) -- ({60:1}) {};
  \path[songblackpath] (0,0) -- ({60:1}) {};
  \path[songblackpath] (1,-1) .. controls (1.5,0) .. ({60:1}) {};
  \path[songblackpathdashed] (2,0.866) -- ({60:1}) {};
  \end{scope}
  \path[songblackpathdashed] (2,0.866) -- (3, 0.866) {};
  \node[songbluenode, label=below right:{$x$}] at (2,0.866) {};
  \node[songbluenode, label=below right:{$y$}] at (3,0.866) {};

  \node[songbluenode, label=above left:{$v_1$}] at (0,0) {};
  \node[songbluenode, label=above right:{$v_2$}] at (1,0) {};
  \node[songbluenode, label=below:{$v_3$}] at (1,-1) {};
  \node[songbluenode, label=below:{$v_4$}] at (0,-1) {};
  \begin{scope}[shift={(0,0)}]
  \node[songrednode,  label=above:{$w$}] at ({60:1}) {};
  \end{scope}
\end{tikzpicture}
\caption{Configuration in Rule \ref{rule:R7.2.2}.}
\label{fig:rule:R7.2.2}
\end{subfigure}%
~
\begin{subfigure}{0.49\textwidth}
\centering
\begin{tikzpicture}[scale=1., every node/.style={transform shape}]
  \path[songblackpathdashed] (0,0) -- (1,-1) {};
  \path[songblackpathdashed] (0,-1) -- (1,0) {};
  \path[songblackpath]  (0,0) -- (1,0) -- (1,-1) -- (0,-1) -- (0,0) {};
  \begin{scope}[shift={(0,0)}]
  \path[songblackpath] (0,0) -- ({60:1}) {};
  \path[songblackpath] (1,0) -- ({60:1}) {};
  \path[songblackpath] (2,0.866) -- ({60:1}) {};
  \end{scope}
  \path[songblackpath] (2,0.866) -- (3,0.866) {};
  \node[songbluenode, label=below right:{$x$}] at (2,0.866) {};
  \node[songbluenode, label=below right:{$y$}] at (3,0.866) {};

  \node[songbluenode, label=below left:{$v_1$}] at (0,0) {};
  \node[songbluenode, label=below right:{$v_2$}] at (1,0) {};
  \node[songbluenode, label=below:{$v_3$}] at (1,-1) {};
  \node[songbluenode, label=below:{$v_4$}] at (0,-1) {};
  \begin{scope}[shift={(0,0)}]
  \node[songrednode,  label=above:{$w$}] at ({60:1}) {};
  \end{scope}
\end{tikzpicture}
\caption{Configuration in Rule \ref{rule:R7.2.5}.}
\label{fig:rule:R7.2.5}
\end{subfigure}
\caption{Configurations in Rules \ref{rule:R7.2.2}, \ref{rule:R7.2.5}.}
\end{figure}

\begin{lemma}\label{lemma:no-4-cycles}
Assume that Rules \ref{rule:R0} -- \ref{rule:R7.2.5} are not applicable. Then there is no component of $G[V_2]$ that contains a~4-cycle as a~subgraph.
\end{lemma}

\begin{proof}
For contradiction assume that Rules \ref{rule:R0} -- \ref{rule:R7.2.5} are not applicable and there is a~component~$Q$ in $G[V_2]$ that contains a~4-cycle as a~subgraph, label the vertices of the 4-cycle $(v_1,v_2,v_3,v_4)$.
By \autoref{lemma:4cycle_situation}, $Q$ is a~subgraph of $K_4$ and there is exactly one red vertex $w$ connected to~$Q$ and $w$ is connected to at least two vertices in~$Q$ and to some component of $G[V_2]$ other than $Q$.

Let $X = V(Q) \cap N(w)$. As $w$ is connected to at least two vertices in $Q$, $X$ either contains at least one diagonal pair of $Q$, or $X$ is of size 2 and forms a non-diagonal pair. In the case $X$ contains a diagonal pair, \autoref{rule:R7.2.2} applies. In the case $X$ forms a non-diagonal pair, \autoref{rule:R7.2.5} applies.
\end{proof}

\subsection{Dealing with stars in $G[V_2]$}

\begin{lemma}\label{lemma:star_situation}
 Assume that Rules \ref{rule:R0} -- \ref{rule:R7.2.5} are not applicable.
 Suppose that there is a connected component $S$ of $G[V_2]$ which is isomorphic to a star with at least $4$ vertices.
Then there is exactly one red vertex $w$ connected to $S$ and $w$ is connected to all the leaves of $S$.
\end{lemma}

\begin{proof}
Assume that Rules \ref{rule:R0} -- \ref{rule:R7.2.5} are not applicable and there is a~star~$S$ in $G[V_2]$ with at least $4$ vertices.

If there is no $P_5$ that uses vertices from~$S$, then \autoref{rule:R1} is applicable on~$S$.
Hence there are red vertices connected to~$S$.
Suppose there are at least two red vertices connected to~$S$.
If the red vertices are not connected to a~single vertex or a~single edge in~$S$, then \autoref{rule:R2} is applicable, since there is a~$P_5$ that uses at least two red vertices.
So suppose the red vertices are connected to a~single vertex or a~single edge in~$S$.
Then from Lemmata~\ref{lemma:6} and \ref{lemma:6'} we know that those red vertices are not connected to any other vertices outside~$S$.
Consequently, there cannot be a~$P_5$ that uses vertices from~$S$ and again \autoref{rule:R1} is applicable on~$S$.

So suppose that there is a~$P_5$ that uses vertices from~$S$ and there is only one red vertex~$w$ connected to~$S$.
If there is a leaf $l$ of $S$ not connected to $w$, then \autoref{rule:domination} applies to $l$.
\end{proof}

\begin{abranchingrule}\label{rule:R8}
 Let $S$ be a connected component of $G[V_2]$ isomorphic to a star with at least $4$ vertices, $l$ an arbitrary one of its leaves and $s$ its center.
 Then branch on $\langle\, l \mid s \,\rangle$.
\end{abranchingrule}

\begin{correctness}
We have to delete something in~$S$, since there is a~path $P_5 = (l_1,w,l_2,s,l_3)$ for some three leaves $l_1,l_2,l_3$ of $S$.

Since all leaves are twins, from \autoref{lemma:twi} we know that we have to try only one of them. Therefore branching on $\langle\, l \mid s \,\rangle$ is correct.
\end{correctness}

\begin{figure}[t]
\centering
\begin{tikzpicture}[scale=1., every node/.style={transform shape}]
  \path[songblackpath] (0,0) -- ({-20:1}) -- (2*0.9397,0) {};
  \path[songblackpath] (0,0) -- ({20:1}) -- (2*0.9397,0){};
  \path[songblackpath] (0,0) -- ({60:1}) {};
  \path[songblackpathdashed] (0,0) -- ({100:1}) {};
  \path[songblackpath] (2*0.9397,0) -- ({60:1}) {};
  \path[songblackpathdashed] (2*0.9397,0) .. controls (2*0.9397,2) and (2*0.9397-2,2) .. ({100:1}) {};
  \path[songblackpathdashed] (2*0.9397,0) -- (0,0) {};

  \node[songbluenode, label=below:{$s$}] at (0,0) {};
  \node[songbluenode, label=below:{$l_1$}] at ({-20:1}) {};
  \node[songbluenode,  label=above right:{$l_2$}] at ({20:1}) {};
  \node[songbluenode, label=above:{$l_3$}] at ({60:1}) {};
  \node[songbluenodedashed, label=above left:{$l_k$}] at ({100:1}) {};
  \begin{scope}[shift={(2*0.9397,0)}]
  \path[songblackpathdashed] (0,0) -- (1,0) -- (2,0) {};
  \node[songrednode, label=above right:{$w$}] at (0,0) {};
  \node[songbluenodedashed, label=below:{$x$}] at (1,0) {};
  \node[songbluenodedashed, label=below:{$y$}] at (2,0) {};
  \end{scope}
\end{tikzpicture}
\caption{Configuration in Rule \ref{rule:R8}.}
\label{fig:rule:R8}
\end{figure}

\begin{observation}\label{lemma:no-stars}
Assume that Rules \ref{rule:R0} -- \ref{rule:R8} are not applicable. Then there are no stars in $G[V_2]$.
\end{observation}

\begin{proof}
For contradiction assume that Rules \ref{rule:R0} -- \ref{rule:R8} are not applicable and there is a~star~$S$ in $G[V_2]$.
By Lemmata \ref{lemma:no-isolated-vertex}, \ref{lemma:no-isolated-edges}, and \ref{lemma:no-isolated-p3-paths} it has at least $4$ vertices.
Then, by \autoref{lemma:star_situation}, \autoref{rule:R8} is applicable.
\end{proof}

\subsection{Dealing with stars with a~triangle in $G[V_2]$}
\begin{lemma}\label{lemma:star_triangle_situation}
 Assume that Rules \ref{rule:R0} -- \ref{rule:R8} are not applicable.
 Suppose that there is a connected component $S^\triangle$ of $G[V_2]$ which is isomorphic to a star with a triangle with at least $4$ vertices.
Then there is exactly one red vertex $w$ connected to $S^\triangle$ and $w$ is connected to all the leaves of $S^\triangle$ and at least one of its triangle vertices.
\end{lemma}

\begin{proof}
Assume that Rules \ref{rule:R0} -- \ref{rule:R8} are not applicable and there is a~star with a triangle~$S^\triangle$ in $G[V_2]$ with at least $4$ vertices.

If there is no $P_5$ that uses vertices from~$S^\triangle$, then \autoref{rule:R1} is applicable on~$S^\triangle$.
Hence there are red vertices connected to~$S^\triangle$.
If there is no red vertex connected to any of the triangle vertices of $S^\triangle$, then \autoref{rule:domination} applies to both the triangle vertices.

So suppose that there is a red vertex connected to some of the triangle vertices, label it $w$.
If there is a leaf $l$ of $S^\triangle$ not connected to a red vertex, then \autoref{rule:domination} applies to $l$.
If the center or some leaf is connected to a different red vertex than $w$, then there is a $P_5$ that uses at least two red vertices and \autoref{rule:R2} applies.
Hence $w$ is the red vertex connected to all the leaves and if there is a red vertex connected to the center, then it is also $w$.
It follows, that there is no red different than $w$ connected to the triangle vertices of $S^\triangle$, as otherwise
there is a $P_5$ that uses at least two red vertices and \autoref{rule:R2} applies.
Therefore, $w$ is the only red vertex connected to $S^\triangle$.
\end{proof}

\begin{lemma}\label{lemma:15}
Assume that Rules \ref{rule:R0} -- \ref{rule:R8} are not applicable.
Let~$S^\triangle$ be a~star with a~triangle in $G[V_2]$, let $s$ be its a~center and $t_1, t_2$ its triangle vertices.
Let~$w$ be a~red vertex connected to~$S^\triangle$ such that $t_1$ is connected to $w$ and $t_2$ is not connected to $w$ (see \autoref{fig:rule:R9}).
If~$F$ is a~solution that contains $t_2$, then at least one of the following holds:
\begin{enumerate}
  \item[(1)] $t_1 \in F$,
  \item[(2)] $F' = (F \setminus \{t_2\}) \cup \{t_1\}$ is a~solution.
\end{enumerate}
\end{lemma}

\begin{proof}
If there is no $P_5$ containing $t_2$, then \emph{(2)} trivially holds. Suppose that every $P_5$ that contains $t_2$ also contains $t_1$, then again \emph{(2)} trivially holds. So assume that there is a~$P_5$ labeled~$P$ that contains $t_2$ but does not contain $t_1$. If for each such~$P$ there is some vertex~$x$ such that $x \ne t_2$ and $x \in F$, then \emph{(2)} holds, since $t_2$ is not needed in the solution. Finally assume that $V(P) \cap F = \{t_2\}$, then, since~$P$ does not contain $t_1$,~$P$ must start at $t_2$ and $P = (t_2,d,p_1,p_2,p_3)$. But then there also exists a~path $P' = (t_1, s,p_2,p_3,p_4)$ and $P'$ is not hit, which is a~contradiction with~$F$ being a~solution and \emph{(1)} must hold.
\end{proof}

\begin{abranchingrule}\label{rule:R9}
Let~$S^\triangle$ be a~star with a~triangle in $G[V_2]$ with at least $4$ vertices, let $s$ be its center, $t_1, t_2$ its triangle vertices and $L$ the set of leaves.
Let~$w$ be a~red vertex connected to~$S^\triangle$ such that $|\{t_1,t_2\} \cap N(w)| \ge 1$, assume that~$w$ is connected to $t_1$ (see \autoref{fig:rule:R9}).
Then branch on $\langle\, t_1 \mid s \mid L \,\rangle$.
\end{abranchingrule}

\begin{correctness}
We have to delete something in $S^\triangle$ since there is a $P_5 = (w,t_1,t_2,s,l)$ for some leaf $l$ of $S^\triangle$. If we do not delete any vertex from $\{s,t_1,t_2\}$, then the only thing we can do is to delete each vertex in~$L$.

So assume that we did not delete all vertices from~$L$, label some remaining vertex from~$L$ as~$x$. If we do not delete anything in $\{t_1, t_2\}$, then we have to delete~$s$, otherwise a~path $(w,t_1,t_2,s,x)$ would remain.

Finally, if $\{t_1,t_2\} \subseteq N(w)$, then $t_1$ and $t_2$ are twins and from \autoref{lemma:twi} we know that we have to try only one of them.
If $t_2$ is not connected to $w$, then from \autoref{lemma:15} we see that deleting only $t_1$ is sufficient. Thus branching on $\langle\, t_1 \mid s \mid L \,\rangle$ is correct.
\end{correctness}

\begin{figure}[t]
\centering
\begin{subfigure}[t]{0.49\textwidth}
\centering
\begin{tikzpicture}[scale=1., every node/.style={transform shape}]
  \path[songblackpath] (0,0) -- ({-100:1}) {};
  \path[songblackpath] (2*0.9397,0) .. controls (2*0.9397,-2) and (2*0.9397-2,-2) .. ({-100:1}) {};
  \path[songblackpath] (0,0) -- ({-60:1}) -- (2*0.9397,0) {};
  \path[songblackpath] ({-100:1}) -- ({-60:1}) {};
  \path[songblackpath] (0,0) -- ({60:1}) {};
  \path[songblackpathdashed] (0,0) -- ({100:1}) {};
  \path[songblackpathdashed] (2*0.9397,0) -- ({60:1}) {};
  \path[songblackpathdashed] (2*0.9397,0) .. controls (2*0.9397,2) and (2*0.9397-2,2) .. ({100:1}) {};
  \path[songblackpathdashed] (2*0.9397,0) -- (0,0) {};

  \node[songbluenode, label=below left:{$s$}] at (0,0) {};
  \node[songbluenode, label=below left:{$t_2$}] at ({-100:1}) {};
  \node[songbluenode,  label=below:{$t_1$}] at ({-60:1}) {};
  \node[songbluenode, label=above:{$l_1$}] at ({60:1}) {};
  \node[songbluenodedashed, label=above left:{$l_k$}] at ({100:1}) {};
  \begin{scope}[shift={(2*0.9397,0)}]
  \path[songblackpathdashed] (0,0) -- (1,0) -- (2,0) {};
  \node[songrednode, label=above right:{$w$}] at (0,0) {};
  \node[songbluenodedashed, label=below:{$x$}] at (1,0) {};
  \node[songbluenodedashed, label=below:{$y$}] at (2,0) {};
  \end{scope}

\end{tikzpicture}
\end{subfigure}
\begin{subfigure}[t]{0.49\textwidth}
\centering
\begin{tikzpicture}[scale=1., every node/.style={transform shape}]
  \path[songblackpath] (0,0) -- ({-100:1})  {};
  \path[songblackpath] (0,0) -- ({-60:1}) -- (2*0.9397,0) {};

  \path[songblackpath] ({-100:1}) -- ({-60:1}) {};
  \path[songblackpath] (0,0) -- ({60:1}) {};
  \path[songblackpathdashed] (0,0) -- ({100:1}) {};
  \path[songblackpathdashed] (2*0.9397,0) -- ({60:1}) {};
  \path[songblackpathdashed] (2*0.9397,0) .. controls (2*0.9397,2) and (2*0.9397-2,2) .. ({100:1}) {};
  \path[songblackpathdashed] (2*0.9397,0) -- (0,0) {};

  \node[songbluenode, label=below left:{$s$}] at (0,0) {};
  \node[songbluenode, label=below left:{$t_2$}] at ({-100:1}) {};
  \node[songbluenode,  label=below:{$t_1$}] at ({-60:1}) {};
  \node[songbluenode, label=above:{$l_1$}] at ({60:1}) {};
  \node[songbluenodedashed, label=above left:{$l_k$}] at ({100:1}) {};
  \begin{scope}[shift={(2*0.9397,0)}]
  \path[songblackpathdashed] (0,0) -- (1,0) -- (2,0) {};
  \node[songrednode, label=above right:{$w$}] at (0,0) {};
  \node[songbluenodedashed, label=below:{$x$}] at (1,0) {};
  \node[songbluenodedashed, label=below:{$y$}] at (2,0) {};
  \end{scope}

\end{tikzpicture}
\end{subfigure}
\caption{Configurations in Rule \ref{rule:R9}.}
\label{fig:rule:R9}
\end{figure}

\begin{observation}\label{lemma:no-stars-with-a-triangle}
Assume that Rules \ref{rule:R0} -- \ref{rule:R9} are not applicable. Then there are no stars with a~triangle in $G[V_2]$.
\end{observation}

\begin{proof}
For contradiction assume that Rules \ref{rule:R0} -- \ref{rule:R9} are not applicable and there is a~star with a~triangle $S^\triangle$ in $G[V_2]$.
By \autoref{lemma:no-isolated-triangles} it has at least $4$ vertices.
Then, by \autoref{lemma:star_triangle_situation}, \autoref{rule:R9} is applicable.
\end{proof}

\subsection{Dealing with di-stars in $G[V_2]$}
\begin{lemma}\label{lemma:distar_situation}
Assume that Rules \ref{rule:R0} -- \ref{rule:R9} are not applicable.
Let $D$ be a connected component of $G[V_2]$ which is isomorphic to a di-star and let $V(D) = \{s,s'\}\cup L \cup L'$, $|L| \ge 1, |L'| \ge 1$, where $s,s'$ are the centers, and $L$ and $L'$ are the leaves connected to $s$ and $s'$, respectively.
Then there is a red vertex $w$ such that every leaf $l \in L$ is connected to $w$ and a red vertex $w'$ (possibly $w=w'$) such that every leaf $l' \in L'$ is connected to $w'$.
Furthermore,
\begin{itemize}
\item if $w \neq w'$ then neither $w$ nor $w'$ is connected to any of $s$ and $s'$, $w$ is not connected to $l'$ and $w'$ is not connected to $l$;
\item there are no red vertices other than $w$ and $w'$ connected to $D$;
\item if $|L|=|L'|=1$, then both $w$ and $w'$ are also connected to some component of $G[V_2]$ other than $D$.
\end{itemize}
\end{lemma}

\begin{proof}
 If there is an $l \in L$ such that no red vertex is connected to $l$, then \autoref{rule:domination} applies. If there were two leaves $l_1$ and $l_2$ in $L$ such that $l_1$ is connected to red vertex $w_1$ and $l_2$ is connected to red vertex $w_2$, then there is a $P_5 = (w_1,l_1,s,l_2,w_2)$ and \autoref{rule:R2} applies. Hence there is $w$ such that every leaf $l \in L$ is connected to $w$, and, if $|L| \ge 2$, then no other red vertex is connected to leaves in $L$.
 A symmetric argument shows that there is a red vertex $w'$ (possibly $w=w'$) such that every leaf $l' \in L'$ is connected to $w'$, and, if $|L'| \ge 2$, then no other red vertex is connected to leaves in $L'$.

 If $w \neq w'$ and $w'$ was connected to $l \in L$, then we have a $P_5 = (w,l,w',l',s')$ for some $l' \in L'$ and \autoref{rule:R2} applies. A symmetric argument shows that if $w \neq w'$, then $w$ is not connected to any $l' \in L'$.
 Moreover, the same argument shows that if $w=w'$, then no other red vertex is connected to $L \cup L'$.

 If there is some red vertex $w_s \neq w'$ connected to $s$, then we have $P_5 = (w_s,s,s',l',w')$ for some $l' \in L'$ and \autoref{rule:R2} applies. If $w'$ is connected to $s$ and $w \neq w'$, then we have a $P_5 = (w,l,s,w',l')$ for some $l \in L$ and $l' \in L'$ and \autoref{rule:R2} applies. Hence, only $w$ can be connected to $s$ and only in case $w=w'$.
 A symmetric argument shows that the same holds for $s'$.
 In particular, we already showed, that in case $w=w'$ there is no other red vertex connected to $D$.

 Hence, if there is some red vertex $w_1$ other than $w$ and $w'$ connected to $D$, then $w \neq w'$ and $w_1$ is connected to some $l \in L$ or some $l'\in L'$. If $w_1$ is connected to both $l \in L$ and $l'\in L'$, then we have a $P_5 = (w,l,w_1,l',w')$ and \autoref{rule:R2} applies.
 If there is some $w_1 \neq w$ connected to some $l \in L$, then we have $|L|=1$ and by \autoref{lemma:6} all red vertices connected to $l$ are not connected to any other component of $G[V_2]$, in particular, by the previous arguments, they are only connected to $l$. However, then \autoref{rule:domination} applies to $l$ and $s$, as $s$ is contained in any $P_5$ containing $l$. Hence, there are no other vertices connected to $L$.
 The same argument shows that \autoref{rule:domination} applies if $w \neq w'$ and $w$ is not connected to any other component of $G[V_2]$.
 A symmetric argument shows that there are no other vertices connected to $L'$ and if $w \neq w'$, then $w'$ is also connected to some component of $G[V_2]$ other than $D$.
 Hence, there are no other red vertices connected to $D$.

 Finally, assume that $|L|=|L'|=1$, $w=w'$ and $w$ is not connected to any component of $G[V_2]$ other than $D$. Then all vertices of $D$ are contained in the same set of $P_5$'s and \autoref{rule:domination} applies.
\end{proof}

\begin{abranchingrule}\label{rule:R10}
Let~$D$ be a~di-star in $G[V_2]$ and let $V(D) = \{s,s'\}\cup L \cup L'$, $|L| \ge 1, |L'| \ge 1$, where $s,s'$ are the centers, and $L$ and $L'$ are the leaves connected to $s$ and $s'$, respectively.
Let $|L| \ge 2$ or $|L'| \ge 2$.
If $|L| \ge 2$, then let $l_1 \in L$ (see \autoref{fig:rule:R10}), otherwise let $l_1 \in L'$. Branch on $\langle\, l_1 \mid s \mid s' \,\rangle$.
\end{abranchingrule}

\begin{correctness}
Assume that $|L| \ge 2$, the other case is symmetric.
By \autoref{lemma:distar_situation} there is a~red vertex~$w$ connected to all leaves in $L$, let $l_1, l_2 \in L$, $l_1 \neq l_2$ (see \autoref{fig:rule:R10}). Moreover, there is no other red vertex connected to $L$.
We have to delete something in $\{l_1,l_2,s,s'\}$ and since $l_1,l_2$ are twins, from \autoref{lemma:twi} we know that we have to try only one of them, thus branching on $\langle\, l_1 \mid s \mid s' \,\rangle$ is correct.
\end{correctness}

\begin{figure}[t]
\centering
\begin{tikzpicture}[scale=1., every node/.style={transform shape}]
  \path[songblackpath] (0,0) -- (1,0) {};
  \path[songblackpath] (0,0) -- ({100:1}) {};
  \path[songblackpath] (0,0) -- ({140:1}) {};
  \path[songblackpathdashed] (0,0) -- ({180:1}) {};

  \begin{scope}[shift={(1,0)}]
    \path[songblackpath] (0,0) -- ({80:1}) {};
    \path[songblackpathdashed] (0,0) -- ({40:1}) {};
  \end{scope}

  \path[songblackpath] (0.5,2) -- ({100:1}) {};
  \path[songblackpath] (0.5,2) -- ({140:1}) {};
  \path[songblackpathdashed] (0.5,2) .. controls (-1.5,2) and (-1.8,1.2) .. ({180:1}) {};
  \path[songblackpathdashed] (0.5,2) -- (0,0) {};
  \path[songblackpathdashed] (0.5,2) -- (1,0) {};
  \begin{scope}[shift={(1,0)}]
    \path[songblackpathdashed] (-0.5,2) -- ({80:1}) {};
    \path[songblackpathdashed] (-0.5,2) -- ({40:1}) {};
  \end{scope}
  \path[songblackpathdashed] (0.5,2) -- (0.5,3) -- (1.5,3) {};

  \node[songbluenode, label=below:{$s$}] at (0,0) {};
  \node[songbluenode, label=below:{$s'$}] at (1,0) {};
  \node[songbluenode, label=above left:{$l_1$}] at ({100:1}) {};
  \node[songbluenode, label=above left:{$l_2$}] at ({140:1}) {};
  \node[songbluenodedashed,  label=left:{$l_k$}] at ({180:1}) {};
  \begin{scope}[shift={(1,0)}]
  \node[songbluenode, label=above right:{$l'_1$}] at ({80:1}) {};
  \node[songbluenodedashed,  label=right:{$l'_m$}] at ({40:1}) {};
  \end{scope}

  \node[songrednode, label=above left:{$w$}] at (0.5,2) {};
  \node[songbluenodedashed, label=above:{$x$}] at (0.5,3) {};
  \node[songbluenodedashed, label=above:{$y$}] at (1.5,3) {};

\end{tikzpicture}
\caption{Configuration in Rule \ref{rule:R10}.}
\label{fig:rule:R10}
\end{figure}

\begin{observation}\label{obs:only_P4s}
Assume that Rules \ref{rule:R0} -- \ref{rule:R10} are not applicable.
Then every di-star in $G[V_2]$ is a $P_4$.
\end{observation}

\begin{proof}
 If there is a di-star which is not a $P_4$, then \autoref{rule:R10} applies.
\end{proof}

\begin{abranchingrule}\label{rule:R12.3.1}
Let~$P_4 = (l,s,s',l')$ be a connected component of $G[V_2]$.
Assume that there is a red vertex $w$ connected to $l$, $l'$, and a~component of $G[V_2]$ other than~$D$, and not connected to any of $s,s'$. Let~$x$ be the vertex~$w$ connects to outside~$D$ and let~$y$ be a~neighbor of~$x$ in $G[V_2]$ (see \autoref{fig:rule:R12.3.1}). Then branch on $\langle\, x \mid y \mid \{l,l'\} \,\rangle$.
\end{abranchingrule}

\begin{correctness}
If none of the vertices $x,y$ is deleted, then we have to delete at least two vertices in~$\{l,s,s',l'\}$ and from \autoref{lemma:delete} we know that we only have to try to delete $\{l,l'\}$. Therefore branching on $\langle\, x \mid y \mid \{l,l'\} \,\rangle$ is correct.
\end{correctness}

\begin{figure}[t]
\centering
\begin{tikzpicture}[scale=1., every node/.style={transform shape}]
  \path[songblackpath] (0,0) -- (1,0) {};
  \path[songblackpath] (0,0) -- ({140:1}) {};

  \begin{scope}[shift={(1,0)}]
    \path[songblackpath] (0,0) -- ({40:1}) {};
  \end{scope}

  \path[songblackpath] (0.5,2) -- ({140:1}) {};

  \begin{scope}[shift={(1,0)}]
    \path[songblackpath] (-0.5,2) -- ({40:1}) {};
  \end{scope}
  \path[songblackpath] (0.5,2) -- (0.5,3) -- (1.5,3) {};

  \node[songbluenode, label=below:{$s$}] at (0,0) {};
  \node[songbluenode, label=below:{$s'$}] at (1,0) {};
  \node[songbluenode, label=above left:{$l$}] at ({140:1}) {};
  \begin{scope}[shift={(1,0)}]
  \node[songbluenode, label=above right:{$l'$}] at ({40:1}) {};
  \end{scope}

  \node[songrednode, label=above left:{$w$}] at (0.5,2) {};
  \node[songbluenode, label=above:{$x$}] at (0.5,3) {};
  \node[songbluenode, label=above:{$y$}] at (1.5,3) {};

\end{tikzpicture}
\caption{Configuration in Rule \ref{rule:R12.3.1}.}
\label{fig:rule:R12.3.1}
\end{figure}

\begin{abranchingrule}\label{rule:R11.1.3}
Let~$P_4 = (l,s,s',l')$ be a connected component of $G[V_2]$.
Assume that there is a red vertex $w$ connected to $l, l'$, a~component of $G[V_2]$ other than~$D$, and at least one of $s,s'$.
Let~$x$ be the vertex~$w$ connects to outside~$D$ and let~$y$ be a~neighbor of~$x$ in $G[V_2]$ (see \autoref{fig:rule:R11.1.3}).
Then branch on $\langle\, x \mid y \mid \{l,s'\} \mid \{s,l'\} \mid \{s,s'\} \,\rangle$.
\end{abranchingrule}

\begin{correctness}
If none of the vertices $x,y$ is deleted, then we have to delete at least two vertices in~$D$. Assume that we want to delete only two vertices in~$D$. Out of six possible pairs of vertices only $\{l,s'\}, \{s,l'\}, \{s,s'\}$ lead to a~solution. Deleting more than two vertices in~$D$ also deletes at least one of the pairs $\{l,s'\}, \{s,l'\}, \{s,s'\}$. Thus branching on $\langle\, x \mid y \mid \{l_1,s'\} \mid \{s,l'_1\} \mid \{s,s'\} \,\rangle$ is correct.
\end{correctness}

\begin{figure}[t]
\centering
\begin{tikzpicture}[scale=1., every node/.style={transform shape}]
  \path[songblackpath] (0,0) -- (1,0) {};
  \path[songblackpath] (0,0) -- ({100:1}) {};

  \begin{scope}[shift={(1,0)}]
    \path[songblackpath] (0,0) -- ({80:1}) {};
  \end{scope}

  \path[songblackpath] (0.5,2) -- ({100:1}) {};
  \path[songblackpath] (0.5,2) -- (0,0) {};
  \path[songblackpathdashed] (0.5,2) -- (1,0) {};

  \begin{scope}[shift={(1,0)}]
    \path[songblackpath] (-0.5,2) -- ({80:1}) {};
  \end{scope}
  \path[songblackpath] (0.5,2) -- (0.5,3) -- (1.5,3) {};

  \node[songbluenode, label=below:{$s$}] at (0,0) {};
  \node[songbluenode, label=below:{$s'$}] at (1,0) {};
  \node[songbluenode, label=above left:{$l$}] at ({100:1}) {};
  \begin{scope}[shift={(1,0)}]
  \node[songbluenode, label=above right:{$l'$}] at ({80:1}) {};
  \end{scope}

  \node[songrednode, label=above left:{$w$}] at (0.5,2) {};
  \node[songbluenode, label=above:{$x$}] at (0.5,3) {};
  \node[songbluenode, label=above:{$y$}] at (1.5,3) {};

\end{tikzpicture}
\caption{Configuration in Rule \ref{rule:R11.1.3}.}
\label{fig:rule:R11.1.3}
\end{figure}

\begin{lemma}\label{lemma:only_P4s_with_two_red}
Assume that Rules \ref{rule:R0} -- \ref{rule:R11.1.3} are not applicable.
Then every connected component of $G[V_2]$ is $P_4 = (l,s,s',l')$ and there are exactly two red vertices $w,w'$ connected to it.
Vertex $w$ is connected to $l$ and to some other component of $G[V_2]$ but not to $s$, $s'$, and $l'$.
Vertex $w'$ is connected to $l'$ and to some other component of $G[V_2]$ but not to $l$, $s$, and $s'$.
\end{lemma}

\begin{proof}
From \autoref{lemma:p5-free-graphs} together with Observations~\ref{lemma:no-isolated-vertex}, \ref{lemma:no-isolated-edges}, \ref{lemma:no-isolated-p3-paths}, \ref{lemma:no-isolated-triangles}, \ref{lemma:no-stars}, and \ref{lemma:no-stars-with-a-triangle} and \autoref{lemma:no-4-cycles}, all components of $G[V_2]$ are di-stars.
By \autoref{obs:only_P4s}, all di-stars are actually $P_4$'s.

Let $P_4 = (l,s,s',l')$ be one such component.
By \autoref{lemma:distar_situation}, there are at most two red vertices connected to it, $w$ connected to $l$ and $w'$ connected to $l'$. Moreover, both $w$ and $w'$ are also connected to some other component of $G[V_2]$.
If $w=w'$ and $w$ is connected to neither $s$ nor $s'$, then \autoref{rule:R12.3.1} applies.
If $w=w'$ and $w$ is connected to at least one of $s$ nor $s'$, then \autoref{rule:R11.1.3} applies.
Hence, $w \neq w'$ and by \autoref{lemma:distar_situation} neither $w$ nor $w'$ is connected to any of $s$ and $s'$, $w$ is not connected to $l'$ and $w'$ is not connected to $l$.
\end{proof}

\begin{abranchingrule}\label{rule:R18}
Let~$P_4 = (l,s,s',l')$ be a connected component of $G[V_2]$.
Let there be two red vertices $w, w'$ connected to leaves $l \in L$ and $l' \in L'$, respectively, and both $s,s'$ have degree exactly two in $G$ (see \autoref{fig:rule:R18}). Then branch on $\langle\, l \mid l' \,\rangle$.
\end{abranchingrule}

\begin{correctness}
By \autoref{lemma:only_P4s_with_two_red} each connected component of $G[V_2]$ is a $P_4$ with two red vertices connected.
Let~$F$ be a~solution. Label the di-star components of $G[V_2]$ as $D_1,D_2,\ldots,D_r$. Observe that~$F$ deletes at least one vertex in each di-star component $D_i$.

Firstly, we construct a~directed graph $G'$ such that $V(G') = V_1$ and there is an edge $e_i = (x,y)$ in $G'$ if and only if~$F$ deletes exactly one vertex in $D_i$ and the deleted vertex is either $s_i^y$ or $l_i^y$ where $l_i^y$ is a~leaf~$y$ connects to in $D_i$ and $s_i^y$ is the~center of $D_i$ to which $l_i^y$ is connected.

We claim that each vertex in $G'$ has outdegree at most one. Indeed, for contradiction assume that vertex~$w$ has outdegree at least two, which means that there are two di-star components $D_i, D_j$ connected to~$w$ such that~$F$ does not contain the leaves~$w$ is connected to in $D_i, D_j$, let them be $l_i^w$, $l_j^w$ and the centers to which these leaves are connected, let them be $s_i^w$, $s_j^w$, respectively. But that implies a~$P_5 = (s_i^w, l_i^w, w, l_j^w, s_j^w)$ in~$G$ and~$F$ would not be a~solution, which is a~contradiction.

Secondly, we construct a~set $F'$ in the following way: (1) for each di-star component $D_i$ where~$F$ deletes at least two vertices, add to $F'$ the two leaves of $D_i$ and (2) for each edge $e_j = (x,y)$ in $G'$ add to $F'$ a~leaf connected to~$y$ in $D_j$.

Finally, $F'$ is also a~solution because in the di-star $D_i$ where~$F$ deleted at least two vertices we know from \autoref{lemma:delete} that it suffices to delete only the leaves of $D_i$ and we claim that in the graph $G \setminus F'$ there is no $P_5$. Indeed, for contradiction assume that there is a~$P_5$ in $G \setminus F'$. But that could only happen if there was a~vertex~$w$ in $G'$ with outdegree at least two, which is a~contradiction.

Therefore $F'$ is a~solution that uses only leaves of the di-stars in~$G$ and from construction of $G'$ and $F'$ we have that $|F'| \le |F|$. Thus branching on $\langle\, l_1 \mid l'_1 \,\rangle$ is correct.
\end{correctness}

\begin{figure}[t]
\centering
\begin{tikzpicture}[scale=1., every node/.style={transform shape}]
  \path[songblackpath] (0,0) -- (1,0) {};
  \path[songblackpath] (0,0) -- ({140:1}) {};

  \begin{scope}[shift={(1,0)}]
    \path[songblackpath] (0,0) -- ({40:1}) {};
  \end{scope}

  \path[songblackpath] (0,2) -- ({140:1}) {};

  \begin{scope}[shift={(1,0)}]
    \path[songblackpath] (0,2) -- ({40:1}) {};
  \end{scope}
  \path[songblackpath] (1,2) -- (2,1.6) {};
  \path[songblackpathdashed] (1,2) -- (2,2.4) {};
  \path[songblackpath] (0,2) -- (-1,1.6) {};
  \path[songblackpathdashed] (0,2) -- (-1,2.4) {};

  \node[songbluenode, label=below:{$s$}] at (0,0) {};
  \node[songbluenode, label=below:{$s'$}] at (1,0) {};
  \node[songbluenode, label=above left:{$l_1$}] at ({140:1}) {};
  \begin{scope}[shift={(1,0)}]
  \node[songbluenode, label=above right:{$l'_1$}] at ({40:1}) {};
  \end{scope}

  \node[songrednode, label=above:{$w_1$}] at (0,2) {};
  \node[songrednode, label=above:{$w_2$}] at (1,2) {};
  \node[anchor=west] at (2,2.4 +0.05) {\textit{di-star}};
  \node[anchor=west] at (2,1.6 +0.05) {\textit{di-star}};
  \node[anchor=east] at (-1,1.6+0.05) {\textit{di-star}};
  \node[anchor=east] at (-1,2.4+0.05) {\textit{di-star}};

\end{tikzpicture}
\caption{Configuration in Rule \ref{rule:R18}.}
\label{fig:rule:R18}
\end{figure}

\begin{observation}\label{lemma:no-di-stars}
 If $G[V_2]$ is non-empty, then at least one of the at least one of Rules \ref{rule:R0} -- \ref{rule:R18} is applicable.
\end{observation}

\begin{proof}
By \autoref{lemma:only_P4s_with_two_red}, if none of Rules \ref{rule:R0} -- \ref{rule:R11.1.3} is applicable, then every connected component of $G[V_2]$ is a~$P_4 = (l,s,s',l')$ with two red vertices $w, w'$ connected to leaves $l \in L$ and $l' \in L'$, respectively, and both $s,s'$ have degree exactly two in $G$. But then \autoref{rule:R18} applies.
\end{proof}

\subsection{Final remarks}
\label{section:final-remarks}

From Observation~\ref{lemma:no-di-stars} we know that there is always at least one rule applicable. It remains to analyze the running time of the disjoint compression routine \textsc{disjoint\_r}.

\begin{theorem}\label{theorem:final2}
The \textsc{disjoint\_r} procedure solves the \textsc{5-PVCwB} problem in $\ostar(3^k)$ time.
\end{theorem}

\begin{proof}
We use the technique of analysis of branching algorithms as described by Fomin and Kratsch~\cite{FominK10}.

Let $T(k)$ be the maximum number of leaves in any search tree of a~problem instance with parameter~$k$. We analyze each branching rule separately and finally use the worst-case bound on the number of leaves over all branching rules to bound the number of leaves in the search tree of the whole procedure.

Let $\langle\, X_1 \mid X_2 \mid \ldots \mid X_l \,\rangle$ be the branching rule to be analyzed. We have that $l \ge 2$ and $|X_i| \ge 1$. This implies the linear recurrence $$T(k) \le T(k - |X_1|) + T(k - |X_2|) + \dotsb + T(k - |X_l|).$$

It is well known that the base solution of such linear recurrence is of the form $T(k) = \lambda^k$ where $\lambda$ is a~complex root of the polynomial $$\lambda^k - \lambda^{k-|X_1|} - \lambda^{k-|X_2|} - \dotsb - \lambda^{k-|X_l|} = 0$$ and the worst-case bound on the number of leaves of the branching rule is given by the unique positive root of the polynomial. This positive root $\lambda$ is called a~\emph{branching factor}.

The worst-case upper bound of the number of leaves in the search tree of the whole procedure is the maximal branching factor among the branching factors of all the branching rules. Be advised that the branching factor does not necessarily correspond to the number of branching calls, e.g., \autoref{rule:R11.1.3} generates 5 branching calls, but 3 of them delete more than one vertex, which results in branching factor of 3 rather than 5. In our case, the worst-case branching factor is 3 (see \autoref{table:branching-factors} for the branching factors), therefore the upper bound of the number of leaves in the search tree is $\ostar(3^k)$.

Now we have to upper bound the number of inner nodes in the search tree. We claim that each path from the root to some leaf of the search tree has at most $\mathcal{O}(|V(G)|)$ vertices. Indeed, each rule removes at least one vertex from~$G$. Therefore the upper bound of the number of inner nodes in the search tree is $\ostar(3^k)$.

Since the running time of each rule (the work that is done in each node of the search tree) is polynomial in $|V(G)|$, we get that the worst-case running time of the whole procedure is $\ostar(3^k)$.
\end{proof}

To understand the key ideas behind iterative compression algorithms and how the \textsc{disjoint\_r} routine is involved, we briefly describe the iterative compression algorithm (for in-depth description see Cygan et al.~\cite[pages 80--81]{CyganFKLMPPS15}).

We start with an empty vertex set $V' = \emptyset$ and empty solution $F = \emptyset$ and work with the graph $G[V']$. Surely, an empty set~$F$ is a~solution for a~currently empty graph $G[V']$. We add vertices $v \in V \setminus V'$ one by one to both $V'$ and~$F$ until $V' = V$ and if at any time the solution becomes too large, i.e. if $|F| = k + 1$, then we start the compression routine.

The compression routine takes~$F$ and goes through every partition of~$F$ into two sets $X,Y$ such that $Y \ne \emptyset$. Here,~$X$ is the part of~$F$ that we want to keep in the solution and~$Y$ is the part of~$F$ that we want to replace with vertices from $V' \setminus F$. Since~$X$ are vertices we already decided to keep in the solution, we remove them from $G[V']$, i.e. we continue with $G' = G[V'] \setminus X$. Now the problem is to find a~solution $F'$ for $G'$ such that $|F'| \le |Y| - 1$ and $F'$ is disjoint from~$Y$. We consider this partition only if $G[Y]$ is $P_5$-free. Indeed, we require that $F'$ is disjoint from~$Y$ so we cannot have any $P_5$ paths in $G[Y]$. To find this smaller disjoint solution $F'$ for $G'$ we use the \emph{disjoint compression routine} which in our case is the \textsc{disjoint\_r} procedure. The smaller solution for $G[V']$ is then constructed as $\hat{F} = X \cup F'$ and it follows from construction of $\hat{F}$ that $|\hat{F}| \le k$.

If after going through all partitions of~$F$ we did not find a~smaller solution for $G[V']$, then we know that~$F$ was optimal in size and signalize that there is no solution.

The complexity of the whole iterative compression algorithm is then computed as follows. The compression routine is called at most  $|V(G)|$ times and the worst case running time of one run of the compression routine can be computed as
 $$\sum_{X \subsetneq F} O^*(3^{k-|X|}) = \sum^{k}_{i=0} \binom{k+1}{i} O^*(3^{k-i}) = O^*(4^k),$$
which finally gives us the following corollary.

\begin{corollary}\label{theorem:final3}
The iterative compression algorithm solves the \textsc{5-PVC} problem and runs in $\ostar(4^k)$ time.
\end{corollary}

\begin{table}[t]
\centering
\caption{Branching factors $\lambda$ of the branching rules.}
\label{table:branching-factors}
\def\arraystretch{1}
\begin{tabular}{|l|l||l|l||l|l|}
  \hline
  Rule & $\lambda$ & Rule & $\lambda$ & Rule & $\lambda$  \\
  \hhline{|=|=||=|=||=|=|}
  \ref{rule:R2}      & 3  & \ref{rule:R7.2.2}  & 3       & \ref{rule:R12.3.1} & 2.415 \\\hline
  \ref{rule:R3}      & 3  & \ref{rule:R7.2.5}  & 2.415   & \ref{rule:R11.1.3} & 3     \\\hline
  \ref{rule:R4}      & 3  & \ref{rule:R8}      & 2       & \ref{rule:R18}     & 2     \\\hline
  \ref{rule:R5}      & 3  & \ref{rule:R9}      & 3       &                    &       \\\hline
  \ref{rule:R6}      & 3  & \ref{rule:R10}     & 3       &                    &       \\\hline
\end{tabular}
\end{table}

\section{Conclusion}
We conclude this paper with a few open questions.

Firstly, we again kindly remind the reader of our recent work on generating efficient algorithms for \textsc{$d$-PVC}~\cite{CervenyS21}. As the generated algorithms consist of thousands (and in some cases hunders of thousands) branching rules, we ask, whether there exist significantly simpler algorithms with comparable running times.


Secondly, as the $d$-\textsc{Hitting Set} algorithm of Fernau~\cite{Fernau10} gets closer to the running time of $\ostar{((d-1)^k)}$ with increasing $d$, we ask, whether one can find a general \textsc{$d$-PVC} algorithm with running time $\ostar((d-1)^k)$ or if it is possible to go below the $d-1$ base of the exponential.

Finally, in our recent work on kernels for \textsc{$d$-PVC} we give a kernel with $O(k^3)$ edges for \textsc{$d$-PVC} and for \textsc{$4$-PVC} and \textsc{$5$-PVC} we have a kernel with $O(k^2)$ edges. As Dell and Melkebeek~\cite{DellM10} have shown that for \textsc{Vertex Cover} it is not possible to achieve a kernel with $O(k^{2-\epsilon})$ edges unless coNP is in NP/poly (which would imply a collapse of the polynomial hierarchy) and as their result extends to \textsc{$d$-PVC}, the best we can hope for is a kernel with $O(k^2)$ edges for \textsc{$d$-PVC}. And therefore we ask, can we bridge the gap between our $O(k^3)$ kernel towards the $O(k^2)$ kernel?


\bibliographystyle{plainurl}
\bibliography{ref}

\end{document}